\documentclass{paper}

\usepackage{amssymb,amsmath,amsthm,amsfonts}
\usepackage [usenames] {xcolor}
\usepackage{enumerate}
\usepackage{cite}
\usepackage{plaatjes}
\usepackage[linenum]{colormath}
\usepackage{fullpage}

\newcommand {\N} {\mathbb {N}}
\newcommand {\eps} {\varepsilon}
\newcommand {\R} {\mathbb {R}}

\newcommand {\eqdef}{:=}

\newcommand {\script} [1] {\ensuremath {\mathcal {#1}}}
\newcommand {\etal} {\textit {et al.}}
\newcommand {\D} {\script {D}}
\DeclareMathOperator{\ch}{ch}
\newcommand {\onion} {\raisebox{-0.25ex}{\includegraphics{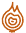}}}

\newtheorem{theorem}{Theorem}[section]
\newtheorem{lemma}[theorem]{Lemma}

\newtheorem{obs}[theorem]{Observation}

\title{
  Unions of Onions: 
   Preprocessing Imprecise Points for Fast Onion
   Decomposition\thanks{A preliminary version appeared as 
   M.~L\"offler and W.~Mulzer. \emph{Unions of Onions: Preprocessing 
   Imprecise Points for Fast Onion Layer Decomposition}.
   Proc. 13th WADS, pp. 487--498, 2013.}
}

\author
{
  Maarten L\"offler\thanks{
    Dept. of Information and Computing Sciences, 
      Universiteit Utrecht, the Netherlands,
    \texttt{m.loffler@uu.nl}}
   \and Wolfgang Mulzer\thanks{
    Institut f\"ur Informatik, Freie Universit\"at Berlin, Germany, 
    \texttt{mulzer@inf.fu-berlin.de}
    }
}

\begin{document}
\maketitle

\begin{abstract}
   Let $\D$ be a set of $n$ pairwise disjoint unit disks in the plane.
   We describe how to build a data structure for $\D$ so that
   for any point set $P$ containing
   exactly one point from each disk, we can quickly find the
   onion decomposition (convex layers) of $P$.

   Our data structure can be built in $O(n \log n)$ time
   and has linear size. Given $P$, we can find its 
   onion decomposition in $O(n \log k)$ time, where $k$ is the number of layers.
   We also provide a matching lower bound.
   
   Our solution is based on a recursive space decomposition,
   combined with a fast algorithm to compute the union of two disjoint onion 
   decompositions.
\end{abstract}

\section {Introduction}

Let $P$ be a planar $n$-point set. Take the convex hull of $P$ and remove
it; repeat until $P$ becomes empty.
This process is called \emph{onion peeling}, and the resulting
decomposition of $P$ into convex polygons is the
\emph{onion decomposition}, or \emph {onion} for short, of $P$.
It can be computed in $O(n \log n)$ time~\cite {c-clps-85}.
Onions provide a natural, more robust, generalization of the convex hull,
and they have applications in pattern recognition, statistics, and 
planar halfspace range searching~\cite{cgl-tpogd-85,h-rsar-72,sf-clntrpdps-99}.

Recently, a new paradigm has emerged for modeling data
imprecision.
Suppose we need to compute some interesting property of
a planar point set. Suppose further that we have some advance
knowledge about the possible locations of the points, e.g.,
from an imprecise sensor measurement. We would like to preprocess
this information, so that once the precise inputs are available,
we can obtain our structure faster.
We will study  the complexity of computing onions
in this framework.

\subsection{Related Work}

The notion of onion decompositions first appears in the 
computational statistics literature~\cite{h-rsar-72}, and several rather 
brute-force algorithms to compute it have been suggested 
(see \cite{e-chp-82} and the references therein).
In the computational geometry community, Overmars and van 
Leeuwen~\cite{ol-mcp-81} presented the first near-linear time algorithm,  
requiring $O(n \log^2 n)$ time.
Chazelle~\cite{c-clps-85} improved this to an optimal $O(n \log n)$ time 
algorithm.
Nielsen~\cite{n-ospcml-96} gave an output-sensitive algorithm to compute 
only the outermost $k$ layers in $O(n \log h_k)$ time, where $h_k$ is the 
number of vertices participating on the outermost $k$ layers.
In $\R^3$, Chan~\cite{c-addsf3cha2nnq-10} described an $O (n \log^{6} n)$ 
expected time algorithm.

The framework for preprocessing regions that represent points was first 
introduced by Held and Mitchell~\cite {hm-tipps-08}, who show how to store 
a set of disjoint unit disks in a data structure such that any point set 
containing one point from each disk can be triangulated in linear time.
This result was later extended to arbitrary disjoint regions in the plane 
by van Kreveld~\etal~\cite{klm-pipast-10}.
L\"offler and Snoeyink first showed that the Delaunay triangulation (or its 
dual, the Voronoi diagram) can also be computed in linear time after 
preprocessing a set of disjoint unit disks~\cite{ls-dtoipiltap-10}.
This result was later extended by Buchin~\etal~\cite{blmm-pipfdtsae-11},
and Devillers gives a practical alternative~\cite{d-dtoippaagafqt-11}.
Ezra and Mulzer~\cite{em-choipitap-13} show how to preprocess a set of 
lines in the plane such that the convex hull of a set of points with 
one point on each line can be computed faster than $n \log n$ time.

These results also relate to the \emph{update complexity} model. 
In this paradigm, the input values or points come with some uncertainty,
but it is assumed that during the execution of the algorithm, the values or 
locations can be obtained exactly, or with increased precision, at a 
certain cost. 
The goal is then to compute a certain combinatorial property or structure 
of the precise set of points, while minimising the cost of the updates 
made by the 
algorithm~\cite{bhkr-eusgcu-05,fggt-chapa-94,hekmr-cmstu-08,tk-itet-11}.

\tweeplaatjes [scale=.99] {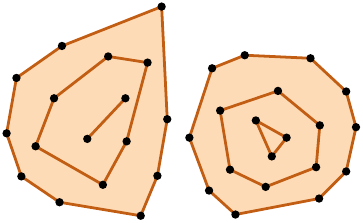} {intro-example-out} 
{ (a) Two disjoint onions. (b) Their union.}

\subsection{Results}

We begin by showing that the union of two disjoint onions can be computed in 
$O (n + k^2 \log n)$ time, where $k$ is the number of layers in the
resulting onion.

We apply this algorithm to obtain an efficient solution to the
onion preprocessing problem mentioned in the introduction.
Given $n$ pairwise disjoint unit disks that model
an imprecise point set, we build a data structure of size
$O(n)$ such that the onion decomposition of an instance can be 
retrieved in $O(n \log k)$ time, where $k$ is 
the number of layers in the resulting onion.
We present several preprocessing algorithms. 
The first is very simple and achieves $O(n \log n)$
expected time. The second and third algorithm make this guarantee 
deterministic, at the cost of worse constants and/or a more involved
algorithm.

We also show that the dependence on $k$ is necessary:
in the worst case, any comparison-based algorithm
can be forced to take $\Omega (n \log k)$ time on some instances.

\section {Preliminaries and Definitions}

Let $P$ be a set of $n$ points in $\R^2$. 
The \emph {onion decomposition}, or \emph {onion}, of $P$, is the 
sequence $\onion(P)$ of nested convex polygons with vertices 
from $P$, 
constructed recursively as follows:
if $P \neq \emptyset$,  we set 
$\onion(P) \eqdef \{\ch(P)\} \cup \onion(P \setminus \ch(P))$,
where $\ch(P)$ is the convex hull of $P$;
if $P = \emptyset$, then 
$\onion(P) \eqdef \emptyset$~\cite {c-clps-85}.
An element of $\onion(P)$ is called a \emph{layer} of $P$.
We represent the layers of $\onion(P)$ as dynamic balanced binary search
trees, so that operations \emph{split} and \emph{join} can be performed
in $O(\log n)$ time.

Let $\D$ be a set of disjoint unit disks in $\R^2$. We say a point set 
$P$ is a \emph {sample} from $\D$ if every disk in $\D$ contains
exactly one point from $P$. 
We write $\log$ for the logarithm with base $2$.

\section {Main Result}

Our data structure and accompanying query algorithm require several pieces, to be described in the
following sections.

\subsection {Unions of Onions}

Suppose we have two point sets $P$ and $Q$, together with their onions. 
We show how to find $\onion(P \cup Q)$ quickly, given that $\onion(P)$
and $\onion(Q)$ are disjoint, given that $\ch(P)$ and $\ch(Q)$ do not overlap.
Deleting points can only decrease
the number of layers, so:

\begin {obs} \label {obs:k}
  Let $P, Q \subseteq \R^2$. 
  Then $\onion(P)$ and $\onion(Q)$ cannot have more layers than 
  $\onion(P \cup Q)$. \hfill$\Box$
\end {obs}

The following lemma constitutes the main ingredient of our onion-union 
algorithm.  A \emph{convex chain} is any connected subset 
of a convex closed curve. 

\begin{lemma} \label{lem:chainhull}
  Let $A$ and $B$ be two non-crossing convex polygonal chains. We can find 
  $\ch(A \cup B)$ in $O (\log n)$ time, where $n$ is the total number
  of vertices in $A$ and $B$.
\end{lemma}

\begin{proof}
  Since $A$ and $B$ do not cross, the pieces of $A$ and $B$ that appear on 
  $\ch(A \cup B)$ are both connected. If not, there would be on $\ch(A \cup B)$ 
  four points that alternate between $A$, $B$, $A$, and $B$, in that order. 
  However, the points on $A$ must be connected inside $\ch(A \cup B)$ by
  the polygonal chain; 
  the same holds for the points on $B$. Thus, the chains $A$ and $B$ would cross, which 
  contradicts the assumption of the lemma. 

  Since $A$ and $B$ are convex chains, we can compute
  $\ch(A), \ch(B)$ in $O(\log n)$ time. Furthermore, since $A$ and $B$
  are disjoint, we can also, in $O(\log n)$ time, make sure that
  $\ch(A) \cap \ch(B) = \emptyset$, by removing parts from $A$ or $B$,
  if necessary.
  Now we can 
  find the bitangents of $\ch(A)$ and $\ch(B)$ in 
  logarithmic time~\cite {ks-cctws-95}.
\end{proof}

\begin{lemma} \label{lem:restore}
  Suppose $\onion(P)$ has $k$ layers. Let $A$ be the outer layer of
  $\onion(P)$, and $p, q$ be two vertices of $A$. 
  Let $A_1$ be the points on $A$ between $p$ and $q$, going counter-clockwise. 
  We can find 
  $\onion (P \setminus A_1)$ in $O (k \log n)$ time.
\end{lemma}

\begin{proof}
  The points $p$ and $q$ partition  $A$ into two pieces, $A_1$ and $A_2$. 
  Let $B$ be the second layer of $\onion(P)$.
  The outer layer of $\onion (P \setminus A_1)$ is the convex hull of 
  $P \setminus A_1$, 
  i.e., the convex hull of $A_2$ and $B$. 
  By Lemma~\ref {lem:chainhull}, we can find it in $O(\log n)$ time.
  Let $p', q' \in P$ be the points on $B$ where the outer 
  layer of $\onion(P \setminus A_1)$ connects.
  We remove the part between $p'$ and $q'$ from $B$, and
  use recursion to compute the remaining layers of $\onion(P \setminus A_1)$ 
  in $O ((k-1) \log n)$ time; see
  Figure~\ref{fig:union-half-eaten+union-cascaded}.
\end{proof}

\tweeplaatjes{union-half-eaten}{union-cascaded} 
{ (a) A half-eaten onion; (b) the restored 
onion.}

\noindent
We conclude with the main theorem of this section:

\begin{theorem}\label{thm:ounion}
  Let $P$ and $Q$ be two planar point sets of total size $n$.
  Suppose that $\onion(P)$ and $\onion(Q)$ are disjoint.
  We can find the onion $\onion (P \cup Q)$ in $O (k^2 \log n)$ 
  time,
  where $k$ is the resulting number of layers.
\end{theorem}

\begin{proof}
  By Observation~\ref{obs:k}, $\onion(P)$ and $\onion(Q)$ each have at most 
  $k$ layers.
  We use Lemma~\ref{lem:chainhull} to find $\ch(P \cup Q)$ 
  in $O(\log n)$ time.
  By Lemma~\ref{lem:restore}, the remainders of 
  $\onion(P)$ and $\onion(Q)$ can be restored to 
  proper onions in $O(k \log n)$ time.
  The result follows by induction.
\end{proof}

\subsection{Space Decomposition Trees}

We now describe how to preprocess the disks in $\D$
for fast divide-and-conquer. 
A \emph{space decomposition tree} (SDT) $T$ 
is a rooted binary tree where each node $v$ is associated with a planar region
$R_v$. The root corresponds to all of $\R^2$; for each leaf 
$v$ of $T$, the region $R_v$ intersects only a constant number of
disks in $\D$.
Furthermore, each inner node $v$ in $T$ is associated with a directed line
$\ell_v$, so that if $u$ is the left child and $w$ the right
child of $v$, then $R_u \eqdef R_v \cap \ell_v^+$ and 
$R_w \eqdef R_v \cap \ell_v^-$. Here, $\ell_v^+$ is the 
halfplane to the left of $\ell_v$ and $\ell_v^-$ the halfplane
to the right of $\ell_v$; see Figure~\ref{fig:dec-tree}.

Let $\alpha, \beta \in (0,1)$, and let $T$ be an SDT.
For a node $v$ of $T$, let $d_v$ denote the number of disks in $\D$ that
intersect $R_v$. We call $T$ an $(\alpha, \beta)$-SDT
for $\D$ if for every inner node $v$ we have that
(i) the line $\ell_v$ intersects at most $d_v^\beta$ disks that intersect
$R_v$; and (ii) $d_u, d_w \leq \alpha d_v$, where $u$ and $w$ are the children
of $v$.

\eenplaatje{dec-tree}{A space decomposition tree for
$21$ unit disks.}

\begin{lemma}\label{lem:tree_complexity}
Let $T$ be an $(\alpha, \beta)$-SDT. The tree $T$ has height $O(\log n)$ and
$O(n)$ nodes.
Furthermore, $\sum_{v\in T} d_v = O(n \log n)$.
\end{lemma}

\begin{proof}
The fact that $T$ has height $O(\log n)$ is immediate from property
(ii) of an $(\alpha, \beta)$-SDT.
For $i = 0, \dots, \log n$, let 
$V_i \eqdef \{v \in T \mid d_v \in [2^i, 2^{i+1})\}$, the set of nodes
whose regions intersect between $2^i$ and $2^{i+1}$ disks.
Note that the sets $V_i$ constitute a partition of the nodes.
Let $\widetilde{V_i} \subseteq V_i$ be the nodes in $V_i$ whose parent
is not in $V_i$. By property (ii) again, the $d_v$ 
along any root-leaf path in $T$ are monotonically decreasing, so the nodes in 
$\widetilde{V_i}$ are unrelated (i.e., no node in $\widetilde{V_i}$ is 
an ancestor or descendant of another node in 
$\widetilde{V_i}$). 
Furthermore, the nodes in $V_i$ induce in $T$ a forest $F_i$ such 
that each tree in $F_i$ has a root from $\widetilde{V_i}$ and 
constant height (depending on $\alpha$).

Let $D_i \eqdef \sum_{v \in \widetilde{V_i}} d_v$.
We claim that for $i = 0, \dots, \log n $, we have  
\begin{equation}\label{equ:boundDi}
D_i \leq n \prod_{j=i}^{\log n}\bigl(1+c2^{j(\beta-1)}\bigr), 
\end{equation}
for some large enough constant $c$.
Indeed, consider a node $v \in \widetilde{V}_j$.
As noted above, $v$ is the root of a tree $F_v$ of constant height in the 
forest induced by $V_j$.
By property (i), any node  $u$ in this subtree adds at most 
$d_u^\beta < 2^{(j+1)\beta}$ additional disk intersections (i.e., 
$d_a +d_b \leq d_u + 2^{(j+1)\beta}$, where $a$, $b$ are the children of
$u$).
Since 
$F_v$ has
constant size, the total increase in disk intersections in $F_v$
is thus at most $c'2^{(j+1)\beta}$, for some constant $c'$. 
Since $d_v \geq 2^j$, it follows that the number of disk intersections
increases multiplicatively by a factor of at most $1+c'2^{(j+1)\beta}/2^j \leq
1+c2^{j(\beta-1)}$, for some constant $c$.
The trees $F_v$ partition $T$ and the root intersects
$n$ disks, so for the nodes in $\widetilde{V_i}$, the total number of 
disk intersections
has increased by a factor of at most 
$\prod_{j=i}^{\log n}\bigl(1+c2^{j(\beta-1)}\bigr)$, giving 
(\ref{equ:boundDi}). The product in (\ref{equ:boundDi}) is easily estimated:
\[
D_i \leq n \prod_{j=i}^{\log n}(1 + c2^{j(\beta-1)})
\leq n e^{\sum_{j=i}^{\log n} c2^{j(\beta-1)}}
= n e^{O(1)} = O(n),
\]
since $\beta < 1$. Hence, each set $\widetilde{V_i}$ has at most
$O(n/2^i)$ nodes for $i = 1, \dots, \log n$. The total size of all 
$\widetilde{V_i}$ is $O(n)$.
Since each $v \in V_i$  lies in a constant size subtree rooted at 
a $w \in \widetilde{V_i}$, it follows that $T$ has $O(n)$ nodes.
For the same reason, we get that
$\sum_{v\in T} d_v = O(n \log n)$.
\end{proof}

Now there are several ways to obtain an $(\alpha, \beta)$-SDT 
for $\D$.
A very simple construction is based on the following lemma, which 
is an algorithmic version of a result by 
Alon~\etal~\cite[Theorem~1.2]{akp-cddsl-89}. 
See Section~\ref{sec:derandomize} for alternative approaches.

\begin{lemma}\label{lem:line-separator}
There exists a constant $c \geq 0$, so that for any set $\D$ of 
$m$ congruent nonoverlapping disks in the plane, there is a 
line $\ell$ with at least $m / 2 - c \sqrt {m \log m}$ disks completely 
to each side of it. We can find $\ell$ in 
$O(m)$ expected time.
\end{lemma}

\begin{proof}
Our proof closely follows Alon~\etal~\cite[Section~2]{akp-cddsl-89}.
Set $r \eqdef \lfloor \sqrt{m/\log m}\rfloor$,
and pick a random integer $z$ between $1$ and $r/2$. 
Find a line $\ell$ whose angle with the $x$-axis is $(z/r)\pi$
and that has $\lfloor m/2 \rfloor $ disk centers on each side.
Given $z$, we can find $\ell$ in $O(m)$ time by a median computation.
The proof by Alon~\etal\ implies that with probability at least $1/2$ over the
choice of $z$, the line $\ell$ intersects at most $c\sqrt{m\log m}$ disks in
$\D$, for some constant $c \geq 0$.
Thus,  
we need two tries in expectation
to find a good line $\ell$. The 
expected running time is $O(m)$. 
\end{proof}

To obtain a $(1/2 + \eps, 1/2 + \eps)$-SDT $T$ for $\D$,  we apply
Lemma~\ref{lem:line-separator} recursively until the region for
each node intersects only a constant number of disks.
Since the expected running time per node is linear
in the number of intersected disks, 
Lemma~\ref{lem:tree_complexity} shows that the total expected running time
is $O(n \log n)$.

By Lemma~\ref{lem:tree_complexity}, the leaves of $T$ induce a planar
subdivision $G_T$ with $O(n)$ faces. We add a large enough
bounding box to $G_T$ and triangulate the resulting graph.
Since $G_T$ is planar, the triangulation has complexity
$O(n)$ and can be computed in the same time (no need for heavy
machinery---all faces of $G_T$ are convex). With each disk
in $\D$, we store the list of triangles that intersect it
(recall that each triangle intersects a constant number of  disks). 
This again takes $O(n)$ time and space.
We conclude with the main theorem of this section:

\begin{theorem}\label{thm:sdt1}
  Let $\D$ be a set of $n$ disjoint unit disks in $\R^2$.
  In $O(n \log n)$ expected time, we can construct
  an $(1/2 + \eps, 1/2 + \eps)$ space decompositon tree $T$ for $\D$. 
  Furthermore, for each disk $D \in \D$, we have
  a list of triangles $T_D$ that cover the leaf
  regions of $T$ that intersect $D$.\qed
\end{theorem}

\subsection{Processing a Precise Input}

Suppose we have an $(\alpha,\beta)$-SDT together with a point location
structure as in Theorem~\ref{thm:sdt1}.
Let $P$ be a sample
from $\D$. Suppose first that we know $k$, the number of 
layers in $\onion(P)$.
For each input point $p_i$, let $D_i \in \D$ be the corresponding disk. 
We check all triangles in $T_{D_i}$,
until we find the one that contains $p_i$.
Since there are $O(n)$ triangles, and each one intersects $O(1)$
disks, this takes $O(n)$ total time
for all points in $P$.
Afterwards, we know for each point in $P$ the leaf of $T$ that contains it.

For each node $v$ of $T$, let $n_v$ be the number of points in the subtree
rooted at $v$. We can compute the $n_v$'s in total time $O(n)$ by a postorder
traversal of $T$.
The \emph{upper tree} $T_u$ of $T$ consists of all nodes $v$ with
$n_v \geq k^2$.
Each leaf of $T_u$ corresponds to a subset of $P$ with
$O(k^2)$ points. For each such subset, we use Chazelle's 
algorithm~\cite{c-clps-85}
to find its onion decomposition in $O(k^2 \log k)$ time. 
Since the subsets are disjoint, this takes $O(n \log k)$ total time.
Now, in order to obtain $\onion(P)$, we perform a postorder traversal of 
$T_u$, using Theorem~\ref{thm:ounion} in each node to unite the onions of 
its children. This gives $\onion(P)$ at the root.

The time for the onion union at a node $v$ is
$O(k^2 \log n_v)$. 
We claim that for $i = 2\log k, \dots, \log n$, the upper tree
$T_u$ contains at most $O(n/2^i)$ nodes $v$ with $n_v \in [2^i, 2^{i+1})$.
Given the claim,  the 
total work is proportional to
\begin{align*}
 & \sum_{v \in T_u} k^2 \log n_v 
 \leq \sum_{i = 2\log k}^{\log n} \frac{n}{2^i} k^2 (i+1) 
 =  nk^2 \sum_{i = 2\log k}^{\log n } \frac{i+1}{2^i}  
  = O(n \log k),
\end{align*}
since the series $\sum_{i=2\log k}^{\log n} (i+1)/2^i$ is dominated by the
first term $(\log k)/k^2$. It remains to prove the claim. Fix 
$i \in \{2\log k, \dots, \log n\}$ and let $V_i$ be the nodes
in $T_u$ with $n_v \in [2_i, 2^{i+1})$, whose parents have
$n_v \geq 2^{i+1}$. Since the nodes in $V_i$ represent disjoint subsets
of $P$, we have $|V_i| \leq n/2^i$. 
Furthermore, by property (i) of an $(\alpha, \beta)$-SDT ,
both children $w_1, w_2$ for every node $v \in T_u$ 
have $n_{w_1}, n_{w_2} \leq \alpha n_v$, so that after $O(1)$ levels,
all descendants $w$ of $v \in V$ have $n_w < 2^i$. The claim follows.

So far, we have assumed that $k$ is given. Using standard exponential
search, this requirement can be removed. More precisely, 
for $i = 1, \dots, \log\log n$, set $k_i = 2^{2^i}$. Run the above algorithm 
for $k = k_0, k_1, \dots$. If the algorithm
succeeds, report the result. If not, abort as soon as it turns
out that an intermediate onion has more than $k_i$ layers and try
$k_{i+1}$. The total time is 
\begin{align*}
  \sum_{i=0}^{\log\log k} O(n 2^i) = O(n \log k), 
\end{align*}
as desired. This finally proves our main result.

\begin {theorem}
  Let $\D$ be a set of $n$ disjoint unit disks in $\R^2$.
  We can build a data structure that stores $\D$, of size $O (n)$, 
  in $O(n \log n)$ expected time, such that given a sample 
  $P$ of $\D$, we can compute 
  $\onion (P)$ in $O (n \log k)$ time, where $k$ is the number of layers in
  $\onion(P)$.\hfill$\Box$
\end {theorem}

\noindent\textbf{Remark.}
Using the same approach, without the exponential search, we can also 
compute the outermost $k$ layers of an onion with arbitrarily many layers 
in $O (n \log k)$ time, for any $k$.
In order to achieve this, we simply abort the union algorithm whenever 
$k$ layers have been found, and note that 
by Observation~\ref {obs:k},
the points in $P$ not on the 
outermost $k$ layers of $\onion (P)$ will never be part of the 
outermost $k$ layers of $\onion(Q)$ for any $Q \supset P$.

\section{Deterministic Preprocessing}\label{sec:derandomize}

We now present alternatives to Lemma~\ref{lem:line-separator}. 
First, we describe a very simple construction that
gives a deterministic way to build an $(9/10 + \eps, 1/2 + \eps)$-SDT
in $O(n \log n)$ time.

\begin{lemma}
Let $\D$ be a set of $m$ non-overlapping unit disks.
Suppose that the centers of $\D$ have been sorted in horizontal and vertical direction.
Then we can find in $O(m)$ time a (vertical or horizontal) line $\ell$, such that
$\ell$ intersects $O(\sqrt m)$ disks and such that $\ell$ has
at least $m/10$ disks from $\D$ completely to each side.
\end{lemma}

\begin{proof}
Let $\D_l$, $\D_r$, $\D_t$, $\D_b$ be the $m/10$ left-, right-,
top-, and bottommost disks in $\D$, respectively.
We can find these disks in $O(m)$ time, since we know the horizontal
and vertical order of their centers. We call 
$\D_o \eqdef \D_l \cup \D_r \cup \D_t \cup \D_b$ the \emph{outer disks}, and
$\D_i \eqdef \D \setminus \D_o$ the \emph{inner disks}.

Let $R$ be the smallest axis-aligned rectangle that contains
all inner disks. Again, $R$ can be found in linear time. 
There are $\Omega(m)$ inner disks, and all
disks are disjoint, so the area of $R$ must be $\Omega(m)$.
Thus, $R$ has width or height $\Omega(\sqrt{m})$; assume 
wlog that it has width $\Omega(\sqrt{m})$. Let $R' \subseteq R$ be
the rectangle obtained by moving the left boundary of $R$ to the right
by two units, and the right boundary of $R$ to the left by two units.
The rectangle $R'$ still has width $\Omega(\sqrt{m})$, and it intersects
no disks from $\D_l \cup \D_r$.
There are $\Omega(\sqrt{m})$ vertical lines that intersect
$R'$ and that are spaced at least one unit apart.
Each such line has at least $m/10$ disks completely to each side, and
each disk is intersected by at most one line. 
Hence, there must be a line that intersects $O(\sqrt{m})$ disks, as claimed.
We can find such a line in $O(m)$ time by sweeping the disks from
left to right.
\end{proof}

The next lemma improves the constants of the previous construction. It
allows us to compute an $(1/2 + \eps, 5/6 + \eps)$-SDT
tree in deterministic time $O(n \log^2 n)$, but it requires comparatively heavy machinery.
\begin{lemma}\label{lem:sdt_matousek}
Let $\D$ be a set of $m$ congruent non-overlapping disks.
In deterministic time $O(m \log m)$, we can find a line $\ell$ such
that there are at least $m/2 - \Theta(m^{5/6})$ disks completely to each side
of $\ell$.
\end{lemma}

\begin{proof}
Let $X$ be a planar $n$-point set, and let $1 \leq r \leq n$ be a parameter.
A \emph{simplicial $r$-partition of $X$} is 
a sequence $\Delta_1, \dots, \Delta_a$ of 
$a = \Theta(r)$ triangles and a partition $X = X_1 \dot\cup \cdots \dot\cup X_a$ of $X$ 
into $a$ pieces such
that  (i) for $i = 1, \dots, a$, we have $X_i \subseteq \Delta_i$
and $|X_i| \in \{n/r, \dots, 2n/r\}$; and (ii) every line $\ell$ intersects 
$O(\sqrt{r})$ triangles $\Delta_i$.
Matou\v{s}ek showed that a simplicial $r$-partition exists for every 
planar $n$-point set and for every $r$. 
Furthermore, this partition can be found in
$O(n \log r)$ time (provided that $r \leq n^{1-\delta}$, for 
some $\delta > 0$)~\cite[Theorem~4.7]{Matousek92}.

Let $\gamma, \delta \in (0,1)$ be two constants to be determined later. Set 
$r \eqdef m^{\gamma}$. Let $Q$ be the set of centers of the disks in $\D$.
We compute a simplicial $r$-partition for $Q$ in $O(m \log m)$ time. 
Let $\Delta_1, \dots, \Delta_a$ be the resulting triangles and
$Q = Q_1 \dot\cup \cdots \dot\cup Q_a$ the partition of $Q$.
Set $s \eqdef m^{\delta}$, and for $i = 1, \dots, s$, let $\ell_i'$
be the line through the origin that  forms an angle $(i/2s)\pi$ with the 
positive $x$-axis.
Let $Y_i$ be the projection of the triangles
$\Delta_1, \dots, \Delta_a$ onto $\ell_i'$.
We interpret $Y_i$ as a set
of weighted intervals, where the weight of an interval is the size $|Q_j|$ of 
the associated point set for the corresponding triangle.
By the properties of the simplicial partition, the interval set $Y_i$ 
has \emph{depth} $O(\sqrt{r})$,
i.e., every point on $\ell_i'$ is covered by at most $O(\sqrt{r})$ intervals 
of $Y_i$.

Note that the sets $Y_i$ can be determined in $O(s r \log r) = 
O(m^{\gamma+\delta} \log m) = O(m)$ total time, for $\gamma, \delta$ small enough.
Now, for each  $Y_i$, we find a point $c_i$ on $\ell_i'$ that has intervals of
total weight $m/2 - O(\sqrt{r}(m/r)) = m/2 - O(m^{1-\gamma/2})$  completely to 
each side. Since the depth of
$Y_i$ is $O(\sqrt{r})$, we can find such a point in time $O(\log r)$ with
binary search, for a total of $O(s \log r) = O(m)$ time (it would even be
permissible to spend time $O(r)$ on each $Y_i$). Let $\ell_i$ be the line
perpendicular to $\ell_{i}'$ through $c_i$.

The analysis of Alon~\etal\ shows that for each $\ell_i$, there are at
most $O(s \log s)$ disks that intersect $\ell_i$ and at least one other line
$\ell_j$~\cite[Section~2]{akp-cddsl-89}. Thus, it suffices to focus on the 
disks in $\D$ that intersect
at most one line $\ell_i$. By simple counting, there is a line
$\ell_i$ that exclusively intersects at most $m/s = m^{1-\delta}$ disks. 
It remains to find such a line in $O(m)$ time. For this, we compute 
the arrangement $\mathcal{A}$ of the strips with width $2$ centered
around each $\ell_i$, together with an efficient point location 
structure. For each cell in the arrangement, we store whether it is covered
by $0$, $1$, or more strips. Using standard techniques, the construction 
takes $O(s^2) = O(m^{2\delta})$ time. 
We locate for each triangle $\Delta_i$ the cells of $\mathcal{A}$
that contain the vertices of $\Delta_i$. This needs 
$O(r \log s) = O(m^{\gamma}\log m)$ steps.
Since every line intersects at most $O(\sqrt{r}) = O(m^{\gamma/2})$ triangles,
we know that there are at most $O(sm^{\gamma/2}) = O(m^{\delta+\gamma/2})$ 
triangles that intersect a cell boundary of $\mathcal{A}$. We call these
triangles the \emph{bad} triangles.

For all other triangles $\Delta_i$, we know that the associated point set
$Q_i$ lies completely in one cell of $\mathcal{A}$. Let $\D_i$ be the 
set of corresponding disks. By using the information stored with the cells, we 
can now determine for each disk $D \in \D_i$ in
$O(1)$ time whether $D$ intersects exactly one line $\ell_i$.
Thus, we can determine in total time $O(m)$ for each line $\ell_i$ the total
number of disks that intersect only $\ell_i$ and whose center is not 
associated with
a bad triangle. Let $\ell$ be the line for which this number is minimum.

In total, it has taken us $O(m\log m)$ steps to find $\ell$. Let us
bound the number of disks that intersect $\ell$. First, we know that
there are at most 
$O(m^{\delta + \gamma/2}\cdot m^{1-\gamma}) = O(m^{1 + \delta - \gamma/2})$
disks whose centers lie in bad triangles. Then, there are at most
$O(m^{\delta}\log m)$ disks that intersect $\ell$ and at least one other line.
Finally, there are at most $m^{1-\delta}$ disks with a center in a good 
triangle that intersect only $\ell$. Thus, if we choose, say, 
$\delta = 1/6$  and $\gamma = 2/3$, 
then $\ell$ crosses at most $O(m^{5/6})$ disks in $\D$. Furthermore, 
by construction, $\ell$ has at least $m/2 - O(m^{2/3})$ disk centers on 
each side. The result follows.
\end{proof}

\noindent\textbf{Remark.} Actually, we can use the approach from 
Lemma~\ref{lem:sdt_matousek} to compute an $(1/2 + \eps, 5/6 + \eps)$-SDT in 
total deterministic time $O(m \log m)$. The bottleneck 
lies in finding the simplicial partition for $Q$. All other steps take $O(m)$ time.
However, when applying Lemma~\ref{lem:sdt_matousek} recursively, we do not need
to compute a simplicial partition from scratch. Instead, as in Matou\v{s}ek's paper,
we can recursively refine the existing partitions in linear 
time~\cite[Corollary~3.5]{Matousek92} (while duplicating the triangles for the
disks that are intersected by $\ell$). Thus, after spending $O(m \log m)$ time
on the simplicial partition for the root, we need only linear time per
node to find the dividing lines, for a total of $O(m \log m)$, by 
Lemma~\ref{lem:tree_complexity}.

\section{Lower Bounds}

\eenplaatje[width=\textwidth] {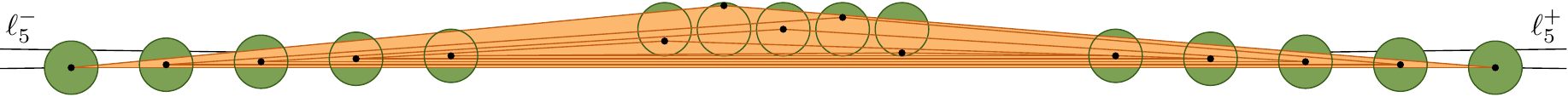} {The lower bound construction 
consists of $n/3$ unit disks centered on a horizontal line 
($5$ in the figure), and two groups of $n/3$ points sufficiently far 
to the left and to the right of the disks.
Distances not to scale.}
\eenplaatje {lowerbound-kgon} {$n/k$ copies of the construction on a 
regular $n/k$-gon. 
}

We now show that 
our algorithm is optimal in the decision tree model. 
The precise nature of the decisions does not matter, 
as long as each decision extracts only a constant number
of bits of information from the input.
We begin with a lower bound of $\Omega(n \log n)$ for $k = \Omega(n)$.
Let $n$ be a multiple of $3$, and consider the lines 
\begin{align*}
\ell_n^- : y = -1/2 -6/n -x/n^2 ;\quad 
\ell_n^+ : y = -1/2 - 6/n +x/n^2 .
\end{align*}
Let $\D_n$ consist 
of $n/3$ disks centered on the $x$-axis at $x$-coordinates 
between $-n/6$ and $n/6$;
a group of $n/3$ disks centered on $\ell_n^-$ at $x$-coordinates 
between $n^2$ and $n^2 + n/3$;
and a symmetric group of $n/3$ disks centered on 
$\ell_n^+$ at $x$-coordinates between $-n^2 - n/3$ and $-n^2$.
Figure~\ref {fig:lowerbound} shows $\D_{15}$.

\begin{lemma}\label{lem:permute}
Let $\pi$ be a permutation on $n/3$ elements. There is a sample $P$ 
of $\D_n$ such that $p_i$ (the point for the $i$th disk from the left in 
the main group) lies 
on layer $\pi(i)$ of $\onion(P)$.
\end{lemma}
\begin{proof}
Take $P$ as the $n/3$ centers of the disks in $\D$ on $\ell_n^-$,
the $n/3$ centers of the disks in $\D$ on $\ell_n^+$, and for each 
disk $D_i \in \D$ on the $x$-axis
the point $p_i = (i - n/6, \pi(i)\cdot 3/n - 1/2)$.
By construction, the outermost layer of $\onion (P)$ contains at 
least the leftmost point on $\ell_n^+$, the rightmost point on 
$\ell_n^-$, and the highest point (with $y$-coordinate $1/2$).
However, it does not contain any more points: the line segments connecting 
these three points have slope at most $2/n^2$. The second highest point 
lies $3/n$ lower, and at most $n/3$ further to the left or the right.
The lemma follows by induction. 
\end{proof}

There are $(n/3)! = 2^{\Theta(n \log n)}$ permutations $\pi$;
so any corresponding decision tree has height
$\Omega(n \log n)$.
We can strengthen the lower bound to 
$\Omega(n\log k)$ by taking $n/k$ copies of $\D_k$
and placing them on the sides of a regular 
$(n/k)$-gon, see Figure~\ref{fig:lowerbound-kgon}. 
By Lemma~\ref{lem:permute}, we can choose independently for
each side of the $(n/k)$-gon one of $(k/3)!$ permutations.
The onion depth will be $k/3$, and the number of permutations
is $((k/3)!)^{n/k} = 2^{\Theta(n \log k)}$.

\begin{theorem}\label{thm:lowerb}
  Let $k \in \N$ and $n \geq k$. There is a set
  $\D$ of $n$ disjoint unit disks in $\R^2$,
  such that any decision-based algorithm to compute $\onion (P)$ for a sample
  $P$ of $\D$,
  based only on prior knowledge of $\D$,
  takes $\Omega (n \log k)$ time in the worst case.
\end{theorem}

The lower bound still applies if
the input points come from an appropriate probability distribution (e.g.,
\cite[Claim~2.2]{Ailon11}). Thus, Yao's minimax 
principle~\cite[Chapter~2.2]{MotwaniRa95} 
yields a corresponding lower bound for any randomized algorithm.

\section{Conclusion and Further Work}

Recently, Hoffmann~\etal~\cite{hkm13} showed how to compute 
in linear deterministic time a line that stabs 
$O(\sqrt{m/(1-2\alpha)})$ disks in a set of $m$ disjoint unit 
disks and has $\alpha m$ centers on each side, for any 
$\alpha < 1/2$. They can also find a line that stabs $O(m^{5/6+\eps})$ 
disks and has exactly $m/2$ centers on each side. 
Using this, one can improve the running times of 
Lemma~\ref{lem:line-separator} and Lemma~\ref{lem:sdt_matousek} 
to linear deterministic time. Note that this does not 
impact the final running time for our original problem.

It would be interesting to understand how much the parameter
$k$ can vary for a set of imprecise bounds and how
to estimate $k$ efficiently.
Further work includes considering more
general regions, such as overlapping disks, disks of
different sizes, or fat
regions. It would also be interesting to consider the
problem in 3D. Three-dimensional onions are not
well understood. The best general algorithm is due to Chan
and needs $O (n \log^{6} n)$ expected time~\cite{c-addsf3cha2nnq-10}, 
giving more room for improvement.

\noindent\textbf{Acknowledgments.}
The authors would like to thank an anonymous reviewer for comments that improved the paper.
M.L.~supported by the Netherlands Organisation for Scientific 
Research (NWO) under grant 639.021.123.  W.M.~supported in part 
by DFG project MU/3501/1.

\bibliographystyle {abbrv}
\bibliography {onions}

\end{document}